\documentclass[11pt]{article}
\usepackage{amsmath,amssymb,amsfonts,amsthm,epsfig}
\usepackage[usenames,dvipsnames]{xcolor}
\usepackage{bm,xspace}
\usepackage{tcolorbox}
\usepackage{cancel}
\usepackage{fullpage}
\usepackage{guy}
\usepackage{framed}
\usepackage{verbatim}
\usepackage{enumitem}
\usepackage{array}
\usepackage{multirow}
\usepackage{afterpage}
\usepackage{mathrsfs}
\usepackage{pifont} 
\usepackage{chngpage}
\usepackage[normalem]{ulem}
\usepackage{boxedminipage}
\usepackage{caption}
\usepackage{forest}
\usepackage{svg}
\usepackage{microtype}
\usepackage{setspace}

\usepackage{algorithm}
\usepackage{algorithmicx}
\usepackage{algpseudocode}

\usepackage{pgfplots}
\pgfplotsset{width=8cm,compat=newest}
\usetikzlibrary{arrows.meta,calc}

\usepackage{ulem}

\def\colorful{1}

\ifnum\colorful=1

\fi
\ifnum\colorful=0

\fi

\makeatletter
\newcommand{\subalign}[1]{%
  \vcenter{%
    \Let@ \restore@math@cr \default@tag
    \baselineskip\fontdimen10 \scriptfont\tw@
    \advance\baselineskip\fontdimen12 \scriptfont\tw@
    \lineskip\thr@@\fontdimen8 \scriptfont\thr@@
    \lineskiplimit\lineskip
    \ialign{\hfil$\m@th\scriptstyle##$&$\m@th\scriptstyle{}##$\hfil\crcr
      #1\crcr
    }%
  }%
}
\makeatother

\newcommand{\Path}{\mathrm{path}}

\newcommand{\Unif}{\mathrm{Unif}}

\newlist{enumprop}{enumerate}{1} 
\setlist[enumprop]{label=\arabic*.,ref=\theproposition.\arabic*}
\crefalias{enumpropi}{proposition}

\newcommand{\hw}{\mathrm{hw}}

\newcommand{\pparagraph}[1]{\bigskip \noindent {\bf {#1}}}

\begin{document}

\title{
The power of quantum circuits in sampling

\vspace{10pt}

}

\newcommand{\authcell}[2]{%
  \begin{tabular}[t]{@{}c@{}}#1\\[6pt]{\slshape #2}\end{tabular}%
}

\author{%
\hspace{2pt}
\begin{tabular*}{\textwidth}{@{\extracolsep{\fill}}ccccc}
  \authcell{Guy Blanc}{Stanford} &
  \authcell{Caleb Koch}{Stanford} &
  \authcell{Jane Lange}{MIT} &
 \hspace{-10pt} \authcell{Carmen Strassle}{Stanford} &
\hspace{-15pt}
  \authcell{Li-Yang Tan}{Stanford}
\end{tabular*}
\vspace{15pt}
}

\date{\small\today}


 \date{\small{\today}}

\maketitle

\begin{abstract}
We give new evidence that quantum circuits are substantially more powerful than classical circuits. We show, relative to a random oracle, that polynomial-size quantum circuits can sample distributions that subexponential-size classical circuits cannot approximate even to TV distance $1-o(1)$. Prior work of Aaronson and Arkhipov (2011) showed such a separation for the case of exact sampling (i.e.~TV distance $0$), but separations for approximate sampling were only known for uniform algorithms.

A key ingredient in our proof is a new hardness amplification lemma for the classical query complexity of the Yamakawa--Zhandry (2022) search problem. We show that the probability that any family of query algorithms collectively finds $k$ distinct solutions decays exponentially in $k$. 
\end{abstract}

  \thispagestyle{empty}
  \newpage 


  \setcounter{page}{1}

\section{Introduction}

A central goal of quantum complexity theory is to understand the relative power of quantum and classical circuits. For circuits computing functions $f : \zo^n\to\zo$ this is the question of whether $\mathsf{BQP/poly} = \mathsf{P/poly}$. In this paper we are interested in the complexity of {\sl distributions}, rather than functions, over $\zo^n$. A standard way of measuring the complexity of a distribution is by the complexity of sampling from it: the size of the smallest circuit that generates samples from it. The most basic question in this setting, as in others, is whether quantumness buys us any appreciable power at all: 

\begin{question}[Does every quantum sampler have a classical counterpart?]
\label{q:main}
     Does every polynomial-size quantum circuit $C$ have a corresponding polynomial-size classical circuit whose output distribution well-approximates that of $C$’s? 
\end{question}

As our main result, we answer~\Cref{q:main} relative to a random oracle: \medskip 

 \begin{tcolorbox}[colback = white,arc=1mm, boxrule=0.25mm]
\begin{theorem} 
\label{thm:main intro}
The following holds relative to a random oracle. For every sufficiently large $n$, there is an explicit distribution $\mathcal{D}^{(n)}$ over $\{0,1\}^n$ that can be sampled exactly by a polynomial-size quantum circuit but no subexponential-size classical circuit can approximate $\mathcal{D}^{(n)}$ even to TV distance $1-2^{-\tilde{\Omega}(\sqrt{n})}$.
\end{theorem}
\end{tcolorbox}\medskip

\Cref{thm:main intro} provides  evidence that quantum circuits are far more powerful than classical ones: In sampling, quantumness can buy us both a dramatic reduction in circuit size as well as a dramatic improvement in accuracy. Furthermore, these advantages hold not only for specific, structured instances (that may have been tailored for a separation), but even for generic, unstructured ones; see~\cite{Bar16} for a broader discussion of random oracle separations. 

The analogue of~\Cref{thm:main intro} for functions remains open: It is not known if $\mathsf{BQP/poly} \ne\mathsf{P/poly}$ relative to a random oracle, and there are formal barriers to such a separation~\cite{AA14}.


\subsection{Related work}

\paragraph{Exact vs.~approximate sampling.} Previously Aaronson and Arkhipov showed~\cite[Theorem 34]{AA11}, under the assumption that the polynomial hierarchy ($\mathsf{PH}$) is infinite, that there is a distribution---the ``complete" boson distribution $\mathcal{D}^{(n)}$---that can be sampled exactly by a polynomial-size quantum circuit but cannot be sampled exactly by polynomial-size classical circuits. 

Our results are incomparable. On one hand,~\cite{AA11}’s holds under the assumption that $\mathsf{PH}$ is infinite whereas~\Cref{thm:main intro} is a relativized statement. On the other hand,~\cite{AA11}’s classical hardness only holds for exact sampling whereas~\Cref{thm:main intro}’s holds for approximate sampling, and even with near-maximal TV distance. Aaronson and Arkhipov pointed to this aspect of their result as its central drawback: ``Why are we so worried about this issue? One obvious reason is that noise, decoherence, photon losses, etc.~will be unavoidable features in any real implementation of a boson computer. As a result, not even the boson computer {\sl itself} can sample exactly from the distribution $\mathcal{D}$! So it seems arbitrary and unfair to require this of a classical simulation algorithm." Therefore, ``we ought to let our simulation algorithm sample from some distribution $\mathcal{D}'$ such that $\| \mathcal{D}-\mathcal{D}'\|\le \varepsilon$ (where $\| \cdot \|$ represents total variation distance), using $\mathrm{poly}(n,1/\varepsilon)$ time." 

There are several obstacles to such an approximate version of~\cite{AA11}'s result. We recap and discuss them in~\Cref{sec:AA obstacles}. We take a different approach and work with a different hard distribution. Our work therefore does not carry any immediate implications for the complexity of boson sampling.

\paragraph{Separations for constant-depth circuits.} Bravyi, Gosset, and K\"onig~\cite{BGK18} raised the analogue of~\Cref{q:main} for  constant-depth circuits (specifically, $\mathsf{QNC}^0$ and $\mathsf{NC}^0$), calling it a challenging open problem. Recent works~\cite{WP23,KOW24} answer this analogue, and in contrast to our result and~\cite{AA11}’s, these separations are unrelativized and unconditional. This mirrors the state of affairs for the circuit complexity of functions, where we have unconditional lower bounds against $\mathsf{NC^0}$ and $\mathsf{AC^0}$~\cite{FSS81,Ajt83,Yao85,Has86} but remain limited to relativized or conditional lower bounds against $\mathsf{P/poly}$.

\paragraph{Dequantizing quantum circuits.} There is a large body of work showing that various subclasses of quantum circuits  {\sl do} have classical counterparts.  Prominent examples include stabilizer circuits (via the Gottesman--Knill theorem~\cite{Got98,AG04}) and matchgate circuits (via Valiant's reduction to Pfaffians~\cite{Val02,TD02}). 

\subsection{Broader context: Representational vs.~algorithmic quantum advantage}

\Cref{q:main} is representational in nature: it asks about the {\sl existence} of classical counterparts of quantum samplers. This stands in contrast to the well-studied algorithmic question of {\sl constructing} a classical counterpart. We now elaborate on the distinction and discuss why separations in the algorithmic setting do not answer~\Cref{q:main}. 

A sampling problem is defined by a map from instances $C\in \{0,1\}^*$ to distributions $\mathcal{D}_{C}$. A classical (resp.~quantum) algorithm that solves the sampling problem receives $C$ as input and  constructs a classical (resp.~quantum) sampler for $\mathcal{D}_C$.\footnote{Alternatively, we can also define the algorithm to be one that receives $C$ as input and generates a  random sample from~$\mathcal{D}_{C}$. These definitions are equivalent by Cook--Levin, for the same reason that $\mathsf{P} = \mathsf{P}\text{-}\mathrm{uniform }\ \mathsf{P/poly}$.} The canonical sampling problem in this context, {\sc Circuit Sampling}, has $C$ being the description of a quantum circuit and $\mathcal{D}_C$ being its output distribution. {\sc Circuit Sampling} trivially admits an efficient quantum algorithm, and the question is whether it admits an efficient classical algorithm. An especially well-studied variant is {\sc Random Circuit Sampling}~\cite{BISBDJBMN18,BFNV19}, the average-case version where $\bC$ is drawn according to a suitable distribution over quantum circuits. 

The complexity of a sampling problem is determined jointly by the circuit complexity of sampling from $\mathcal{D}_C$ and the complexity of {\sl constructing} a sampler for $\mathcal{D}_C$ from $C$. An efficient algorithm for the sampling problem implies the existence of small samplers for $\mathcal{D}_C$ for every instance $C$, but the converse does not necessarily hold. 
Consequently, quantum advantage for  sampling problems conflate the following possibilities: 
\begin{quote}
    \noindent {\bf Possibility 1:} There is a polynomial-size quantum sampler with no equivalent polynomial-size classical counterpart. \ \medskip \ 
    
    \noindent {\bf Possibility 2:} Every polynomial-size quantum sampler has an equivalent polynomial-size classical counterpart, but it is hard to construct such a counterpart efficiently. 
\end{quote}

Both are interesting statements, but ideally we would isolate which is true. Separations such as~\Cref{thm:main intro} and~\cite{AA11}’s result isolate the former, thereby addressing~\Cref{q:main}.

\begin{remark}
Bene Watts and Parham~\cite{WP23} refer to this distinction between algorithmic and representational separations as the distinction between separations for ``input-dependent" problems (defined by a map from instances to distributions for each $n$, as in {\sc Circuit Sampling}) and ``input-independent" ones (defined by single distribution for each $n$, as in~\Cref{thm:main intro} and~\cite{WP23}'s result). 

Quoting the authors, ``While input-dependent problems ask about a classical
system's ability to {\sl process} information, input-independent problems instead study what distributions classical
systems can {\sl prepare}." Put differently, algorithmic separations show that more can be computed efficiently in the quantum world than in the classical  world, whereas representational separations can be seen as showing that the quantum world {\sl itself} is richer and more complex than the classical world. 
\end{remark}

\subsubsection{Three requirements for a representational separation} To establish quantum advantage in the representational setting, one needs the following: 

\medskip 

 \begin{tcolorbox}[colback = white,arc=1mm, boxrule=0.25mm]
\vspace{5pt} 

\begin{enumerate}[leftmargin=1.5em]
    \item {\bf Single distribution for each $n$.} A distribution ensemble $\mathscr{D} = \{ \mathcal{D}^{(n)}\}_{n\in \mathbb{N}}$ where $\mathcal{D}^{(n)}$ is a distribution over $\{0,1\}^n$ that depends only on $n$. For ease of reference we call this the ``one-per-slice" property. 
    
    \item[2.] {\bf Quantum easiness.} For all $n$, there is a polynomial-size quantum circuit that samples $\mathcal{D}^{(n)}$. 
\item[3.] {\bf Classical hardness.} For sufficiently large $n$, no polynomial-size classical circuit can approximately sample $\mathcal{D}^{(n)}$. Equivalently, no polynomial-time {\sl non-uniform} classical algorithm can approximately sample $\mathcal{D}^{(n)}$.
\end{enumerate}
\vspace*{-4pt}
\end{tcolorbox}\medskip

As discussed, sampling problems such as {\sc Circuit Sampling} and {\sc Random Circuit Sampling} do not satisfy the  one-per-slice property.

An important one-per-slice problem from the literature is {\sc Fourier Sampling}~\cite{Aar10,AA15}. This problem is specified by an oracle ensemble $\{ \mathcal{O}^{(n)} : \zo^n\to \zo\}_{n\in \N}$ and $\mathcal{D}^{(n)}$ is the Fourier distribution of $\mathcal{O}^{(n)}$. This problem is quantumly easy with a single query to $\mathcal{O}^{(n)}$. Existing random-oracle lower bounds apply to approximate sampling~\cite{AA15,AC17}, but only rule out classical uniform algorithms (i.e.~Turing machines), not non-uniform ones (i.e.~circuits). These lower bounds therefore again leave open the possibility that the $\mathcal{D}^{(n)}$'s of {\sc Fourier Sampling} all have small classical circuits for all oracle ensembles, and such circuits are just hard to construct. 


\begin{remark}[Salting~\cite{CDGS18}]
Salting is a generic technique for upgrading random-oracle lower bounds against uniform algorithms to ones against non-uniform algorithms. The basic idea is to consider $k$ independent random oracles $\bmcO_1,\ldots, \bmcO_k$ instead of a single one. The algorithm is given as input an index $i\in [k]$ and is  asked to solve the computational problem at hand relative to $\bmcO_i$. It can be shown that if a problem is hard against uniform algorithms, then the “$k$-salted version” for an appropriate choice of $k$ is hard against non-uniform algorithms. 

Applied to {\sc Fourier Sampling}, however, this turns a single hard distribution $\mathcal{D}^{(n)}$ into $k$ many distributions $\mathcal{D}^{(n)}_1,\ldots,\mathcal{D}^{(n)}_k$. The problem that is hard against non-uniform algorithms is: Given $i\in [k]$, output a classical sampler for $\mathcal{D}^{(n)}_i$. This salted problem is now a sampling problem defined by the map $i \mapsto \mathcal{D}^{(n)}_i$ and no longer satisfies the one-per-slice property. We are back to square one: Hardness for this problem leaves open the possibility that all the $\mathcal{D}^{(n)}_i$'s have small classical circuits, and such circuits are just hard to construct even with advice. 
\end{remark}

\begin{remark}[Uniform quantum easiness; non-uniform classical hardness]
Our distribution ensemble in~\Cref{thm:main intro} is {\sl explicit} in the sense that there is a polynomial-time (classical) uniform algorithm that takes as input $1^n$ and outputs a quantum circuit that exactly samples $\mathcal{D}^{(n)}$. Our quantum easiness is therefore witnessed by a uniform algorithm whereas our classical hardness rules out non-uniform algorithms. In other words,~\Cref{thm:main intro} in fact proves  ``$\mathsf{BQP}^{\bmcO}\not\subseteq \mathsf{P/poly}^{\bmcO}$ for distributions", with $\mathsf{BQP}$ instead of $\mathsf{BQP/poly}$.
\end{remark}

\section{Technical Overview}

\paragraph{Quick overview of Yamakawa--Zhandry.}  Our starting point is the work of Yamakawa and Zhandry~\cite{YZ22} who introduced an $\mathsf{NP}$ {\sl search} problem, {\sc NullCodeword},  and a quantum vs.~classical query separation for it.  In {\sc NullCodeword}, $\mathcal{C} \sse \Sigma^n$ is a code (satisfying certain list-recoverable and list-decodable properties) and the algorithm is given query access to an oracle $\mathcal{O}: \Sigma \to\zo$. Its goal is to find a codeword whose coordinates all map to $0$ under $\mathcal{O}$, i.e.~an element of the set: 
\[ \mathcal{C}_{\mathcal{O}} \coloneqq \{ c \in \mathcal{C} \colon \mcO(c_1) = \cdots =\mcO(c_n) = 0\}.\]
\cite{YZ22} proved that {\sc NullCodeword} satisfies:

\begin{description}
    \item[Quantum easiness.] There is a $\poly(n)$-query quantum algorithm that solves {\sc NullCodeword} w.h.p.~relative to a random oracle $\bmcO$. In fact, there is a $\poly(n)$-query quantum algorithm that samples $\mathrm{Unif}(\mcC_{\bmcO})$ w.h.p.~relative to $\bmcO$ (\Cref{thm:YZ quantum easiness}).  
    
\item[Classical hardness.] Every $2^{n^{\Omega(1)}}$-query classical algorithm for {\sc NullCodeword} fails w.h.p.~relative to $\bmcO$. 
\end{description}

\subsection{Overview of our approach}

Setting aside the distinction between search and sampling for now, recall that lower bounds against $q$-query algorithms straightforwardly give  oracle lower bounds against all time-$q$ uniform algorithms, but not all size-$q$ circuits. For the former, it suffices to identify, for every time-$q$ uniform algorithm, a sufficiently large input length $n$ on which it fails. For the latter we need a stronger statement with the order of quantifiers switched: There is a sufficiently large $n$ on which {\sl all} size-$q$ circuits fail.  

One way of achieving the latter is to prove that $q$-query algorithms succeed with {\sl tiny} probability ($\le q^{-\Omega(q)}$), so small that we can afford a union bound over all size-$q$ circuits. {\sc NullCodeword} cannot satisfy this: The trivial algorithm that simply outputs a uniform random element of $\Sigma^n$ succeeds with probability $\ge 2^{-\Theta(n)}$. We want to be able to handle $q$ being any $\poly(n)$, ideally even $2^{n^{\Omega(1)}}$.

We obtain our result in two steps, the first of which is a hardness amplification lemma for {\sc NullCodeword}. We consider a variant that we call {\sl $k$-fold {\sc NullCodeword}} (and denote as $\textsc{NullCodeword}^{\otimes k}$), where the goal is to output $k$ distinct elements of $\mathcal{C}_{\bmcO}$. We show that this is indeed a much harder problem than {\sc NullCodeword}---by choosing $k$ to be sufficiently large, we can drive the success probability of any $q$-query algorithm down to $\le q^{-\Omega(q)}$. We then show how this extreme hardness of $\textsc{NullCodeword}^{\otimes k}$ translates into similar hardness of sampling from the solutions space of {\sc NullCodeword}, i.e.~sampling $\mathrm{Unif}(\mcC_{\bmcO})$. In preserving the tiny $\le q^{-\Omega(q)}$ success probability of $q$-query algorithms, we are able to union bound over all size-$q$ circuits and show that they must all fail at sampling $\mathrm{Unif}(\mcC_{\bmcO})$. This along with the quantum easiness of sampling $\mathrm{Unif}(\mcC_{\bmcO})$ shown in~\cite{YZ22} yields~\Cref{thm:main intro}. We now detail the two steps. 



\subsection{$k$-fold {\sc NullCodeword} and its hardness}

Every search problem admits a $k$-fold version as defined above. However, this version may not be any harder than the original problem itself, since the solutions could be related in a way that having found one makes it easier to find another. Our hardness amplification lemma shows that this is not the case for $\textsc{NullCodeword}$: It formalizes a sense in which having found many  solutions does not help much in finding another. 

As it turns out, for the connection to sampling we have to defeat a computational model that is stronger than the standard query model. Rather than a single query algorithm that outputs a $k$-tuple in $\Sigma$, our query model for solving $\textsc{NullCodeword}^{\otimes k}$ is a {\sl family} $\mathcal{F}$ of $m\gg k$ query algorithms, each of which outputs a single element of $\Sigma$. We say that $\mathcal{F}$ {\sl $k$-succeeds} on an oracle $\mcO$ if the $m$ query algorithms, when each run on $\mcO$, collectively finds at least $k$ distinct solutions (\Cref{def:k-succeeds}). See~\Cref{fig:succeed} for an illustration. \medskip

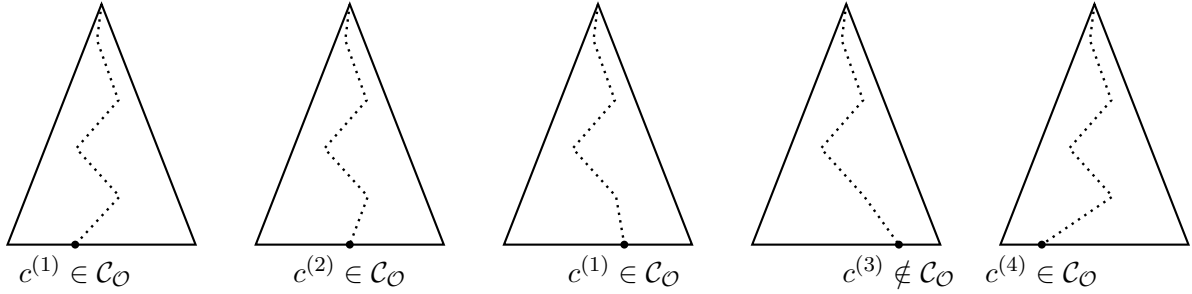
\begin{figure}[h]
    \centering
    \begin{tikzpicture}[x=1cm,y=1cm]

\def\labels{$c^{(1)} \in \mcC_{\mcO}$, $c^{(2)} \in \mcC_{\mcO}$,$c^{(1)} \in \mcC_{\mcO}$, $c^{(3)} \notin \mcC_{\mcO}$, $c^{(4)} \in \mcC_{\mcO}$}

\def\w{2.5}      
\def\h{3.2}      
\def\gap{0.8}    
\def\cs{0.10}    
\def\labeldrop{-0.35}      
\def\treelinewidth{0.8pt}  

\tikzset{
  tree/.style={line width=\treelinewidth},          
  pathzig/.style={dotted, line width=0.9pt},
  leafdot/.style={circle, fill, inner sep=0pt, minimum size=3pt}
}

\newcommand{\crosspoint}[6]{%
  \coordinate (#6L) at ($#1!#4!#3$);%
  \coordinate (#6R) at ($#2!#4!#3$);%
  \pgfmathsetmacro{\alphacoeff}{\cs + (1 - 2*\cs)*#5}%
  \coordinate (#6)  at ($(#6L)!\alphacoeff!(#6R)$);%
}

\newcommand{\IfNonempty}[3]{%
  \if\relax\detokenize{#1}\relax #2\else #3\fi
}

\foreach \lab [count=\i from 1] in \labels {
  \pgfmathsetmacro{\x}{(\i-1)*(\w+\gap)}
  \coordinate (Ai) at (\x,0);
  \coordinate (Bi) at (\x+\w,0);
  \coordinate (Ci) at (\x+0.5*\w,\h);

  \draw[tree] (Ai) -- (Ci) -- (Bi) -- cycle;

  \pgfmathtruncatemacro{\m}{mod(\i,5)}
  \pgfmathsetmacro{\leafq}{0.15 + 0.175*(\m)}
  \pgfmathsetmacro{\leafcoeff}{\cs + (1 - 2*\cs)*\leafq}
  \coordinate (leaf\i) at ($(Ai)!\leafcoeff!(Bi)$);

  \crosspoint{(Ai)}{(Bi)}{(Ci)}{0.85}{0.30}{p\i1}
  \crosspoint{(Ai)}{(Bi)}{(Ci)}{0.60}{0.78}{p\i2}
  \crosspoint{(Ai)}{(Bi)}{(Ci)}{0.40}{0.22}{p\i3}
  \crosspoint{(Ai)}{(Bi)}{(Ci)}{0.20}{0.65}{p\i4}
  \draw[pathzig] (Ci) -- (p\i1) -- (p\i2) -- (p\i3) -- (p\i4) -- (leaf\i);

  \node[leafdot] at (leaf\i) {};

  \IfNonempty{\lab}{}{%
    \node at ($(leaf\i)+(0,\labeldrop)$) {\lab};   
  }
}

\end{tikzpicture}
    \caption{This family of $5$ query algorithms outputs $4$ elements of $\mcC_{\mcO}$, among which $3$ are distinct. It therefore $3$-succeeds on $\mcO$.}
    \label{fig:succeed}
\end{figure}

Our hardness amplification lemma for {\sc NullCodeword} is as follows: \medskip

\begin{tcolorbox}[colback = white,arc=1mm, boxrule=0.25mm]
\begin{lemma}
    \label{lem:hardness amplification intro}
    Let $\mcC$ be an $(\ell,L,t)$-list recoverable code (\Cref{def:list recoverability}) and $\mcF$ be a size-$m$ family of $q$-query algorithms. 
Then 
\[    \Prx_{\bmcO}[\mathcal{F}\text{ }(L+1)\text{-succeeds on }\bmcO ]\le (mL\cdot2^{-(n-t)})^{\floor{\ell/q}}.
    \]
\end{lemma}

\end{tcolorbox}
\medskip

The crux of the proof of~\Cref{lem:hardness amplification intro} is the following definition: We say that an ordered sequence of query algorithms $T^{(1)},\ldots, T^{(r)}$ is {\sl well-spread} on an oracle $\mathcal{O}$ if for every $i \in [r]$,  the output of $T^{(i)}$ on $\mcO$,  denoted $T^{(i)}(\mcO)$, contains many coordinates that are unqueried by 
$T^{(1)},\ldots,T^{(i-1)}$ when run on $\mcO$. See~\Cref{def:well-spread} for the formal definition. Intuitively, this says that $T^{(i)}$ had to do substantial {\sl additional} work to find $T^{(i)}(\mcO)$, on top of the work that $T^{(1)},\ldots, T^{(i-1)}$ already put in to find $T^{(1)}(\mcO),\ldots,T^{(i-1)}(\mcO)$.  With this definition in hand we prove:

\begin{description}
    \item[\Cref{claim:unlikely well spread and all zeroes}:] Any sequence of query algorithms $T^{(1)},\ldots,T^{(r)}$ is extremely unlikely to be well-spread {\sl and} have their outputs to all map to $0^n$ under a random oracle~$\bmcO$ (i.e.~all land in $\mathcal{C}_{\bmcO})$. The probability that both events happen decays exponentially in $r$. Well-spreadness therefore buys us approximate independence.
    \item[\Cref{clm:well-spread-exists}:] If $\mcF$ is a family that $k$-succeeds on $\mcO$, there must exist a long sequence of  $T^{(i)}$'s within $\mcF$ that is well-spread on $\mcO$. 
\end{description}

\Cref{lem:hardness amplification intro} follows easily from~\Cref{claim:unlikely well spread and all zeroes,clm:well-spread-exists}. 

\begin{remark}[Comparison with direct product theorems]
The domain of $\textsc{NullCodeword}^{\otimes k}$ is the same as that of {\sc NullCodeword}: In both cases, the algorithm queries a single oracle~$\mcO$. This stands in contrast to the setting of ``direct product theorems" studied throughout complexity theory, where the $k$-fold direct product of a function $f : X \to Y$ is defined to be $f^{\otimes k} : X^{(1)} \times \cdots \times  X^{(k)} \to Y^{(1)} \times \cdots \times Y^{(k)},$  
\[ f^{\otimes k}(x^{(1)},\ldots,x^{(k)}) = (f(x^{(1)}),\ldots,f(x^{(k)})).\]
We are interested in the complexity of finding $k$ distinct solutions of a single search problem, whereas direct product theorems concern the complexity of solving $k$ independent instances of a problem. 
\end{remark}

\subsection{From $k$-fold {\sc NullCodeword} to sampling}
\label{sec:extract}

Finally, we lift the hardness of $\textsc{NullCodeword}^{\otimes k}$ to the hardness of approximate sampling from the solution space of $\textsc{NullCodeword}$.  Let $\bT$ be a randomized query algorithm that w.h.p.~over a random oracle $\bmcO$ samples a distribution that is non-trivially close to $\mathrm{Unif}(\mcC_{\bmcO})$, i.e.~one satisfying $d_{\mathrm{TV}}(\bT(\bmcO),\Unif(\mcC_{\bmcO})) \le 1-\eps$. We would be done if we can extract from $\bT$ a single family $\mcF$ of $m$ query algorithms that $k$-succeeds on $\bmcO$ with non-negligible probability, since this would contradict~\Cref{lem:hardness amplification intro}. We will not quite show this, but will instead  extract from $\bT$ a small {\sl collection} of families $\mcF^{(1)},\ldots,\mcF^{(u)}$ satisfying: For every $\mcO$ such that $d_{\mathrm{TV}}(\bT(\mcO),\Unif(\mcC_{\mcO})) \le 1-\eps$, there is an $i\in [u]$ such that $\mcF^{(i)}$ $k$-succeeds on $\mcO$. Our setting of parameters will be such that the success probability of~\Cref{lem:hardness amplification intro} is small enough to afford a union bound over the $u$ many families, so this suffices for our purposes.

The existence of this small collection is based on the following probabilistic fact:
\begin{fact}
\label{fact:many distinct}
If $\mcD$ is a distribution that is $1-\eps$ close to $\Unif(H)$, then $m$ independent draws from $\mcD$ will contain $k = \Omega(m\eps)$ many distinct elements of $H$ in expectation. 
\end{fact}

A naive application of~\Cref{fact:many distinct} results in a collection of families that is too large for our purposes. To get around this, we show that~\Cref{fact:many distinct} continues to hold if the $m$ draws are only pairwise independent rather than fully independent (\Cref{claim:pairwise-hits-often}). We use this claim along with known randomness-efficient constructions of pairwise independent random variables.

\begin{remark}
    Aaronson showed an equivalence between search and approximate sampling for the algorithmic setting~\cite{Aar14}, but his techniques do not apply to the representational setting. Our approach, which relates the $k$-fold version of any search problem to the query complexity of sampling its solution space, could be helpful for other problems in the representational setting.
\end{remark}

\subsection{Obstacles in extending~\cite{AA11}'s separation to the approximate setting}
\label{sec:AA obstacles}

Aaronson and Arkhipov obtained their representational separation for exact sampling in two steps: 

\begin{description}
    \item[Step 1:] They first proved an algorithmic separation for the exact {\sc BosonSampling} problem. Here the algorithm is given as input the description of a boson distribution and is asked to generate samples from it. (See~\cite[Section 3]{AA11} for the formal setup.) They showed, under the assumption that $\mathsf{PH}$ is infinite, that no classical algorithm can solve {\sc BosonSampling} exactly in polynomial time.
        \item[Step 2:] They then strengthened this algorithmic separation to a representational one~\cite[Theorem 34]{AA11}. They did so by showing, for each $n\in \N$, the existence of a  specific boson distribution $\mathcal{D}^{(n)}$ that is ``complete" in the sense of being the hardest one (under nondeterministic reductions).  This enabled them to lift the hardness of {\sc BosonSampling} to the representational setting, by considering the one-per-slice version of {\sc BosonSampling} where each slice comprises only of this hardest distribution.  
\end{description}

\cite{AA11} gave two very different proofs of Step 1. The first is based on the $\#{\mathsf{P}}$-completeness of {\sc Permanent}~\cite{Val79} and the second on a theorem of Knill, Laflamme, and Milburn~\cite{KLM01}. It is the second proof that they build on for Step 2. 

The bulk of~\cite{AA11}'s paper laid out a detailed program, based on two still-unproven conjectures---the Permanent of Gaussians Conjecture and the Permanent Anti-Concentration Conjecture---for extending their first proof of Step 1, the one based on {\sc Permanent}, to the approximate setting. They remarked that they ``do not know how to generalize the second proof to say anything about the hardness of approximate sampling". Therefore, even assuming both conjectures, their approach does not give a representational separation for approximate sampling. 

\section{Discussion}

Many problems in areas spanning computer science, statistics, and machine learning are distributional in nature, where algorithms are expected to succeed on instances drawn according to an unknown distribution~$\mcD$. Naturally, we would like to make as mild an assumption regarding $\mcD$ as possible, the mildest one arguably being just that  $\mcD$ is samplable by a polynomial-size circuit. 

Separation such as ours give formal evidence that the class $\mathscr{D}_{\mathrm{quantum}}$ of distributions samplable by polynomial-size {\sl quantum} circuits  is far richer than the corresponding classical class $\mathscr{D}_{\mathrm{classical}}$. Can these separations be used to show that certain distributional problems become much harder when the algorithm is expected to succeed with respect all $\mcD\in \mathscr{D}_{\mathrm{quantum}}$  instead of $\mathscr{D}_{\mathrm{classical}}$? That is, are these problems much harder in the quantum world compared to the classical one?

\section{Preliminaries}

\paragraph{Basic notation and writing conventions. }
We write $[n]$ to denote the set $\{1,2,\ldots,n\}$. We use \textbf{boldface} letters, e.g.~$\bx,\bmcO$, to denote random variables. We write $\Unif(S)$ to denote the uniform distribution over the set $S$. Given a distribution $\mcD$ over $\zo^n$ and a point $x\in \zo^n$, we let $\mcD(x)\coloneqq \Prx_{\by\sim \mcD}[x=\by]$.
\begin{definition}[TV distance]
    \label[definition]{def:TV}
    For any distributions $\mcD_1$ and $\mcD_2$ over the same finite domain $X$, the \emph{total variation distance (TV distance)} between $\mcD_1$ and $\mcD_2$ is the quantity
    \begin{equation*}
d_{\mathrm{TV}}(\mcD_1,\mcD_2) =  \frac{1}{2}\cdot \sum_{x \in X} \abs*{\mcD_1(x) - \mcD_2(x)}.
    \end{equation*}
\end{definition}

\paragraph{Codes.} We work with codes $\mcC$ over the domain $\Sigma_1 \times \cdots\times \Sigma_n$ where the alphabets $\Sigma_i$'s are disjoint. 
We write $\Sigma \coloneqq \Sigma_1 \cup \cdots \cup \Sigma_n$ to denote the set of all possible symbols.

\begin{definition}[Codewords that are $t$-queried by $S$]
\label{def:t-queried}
    Let $\mathcal{C}\sse \Sigma_1\times\cdots\times \Sigma_n$ be a code and $S\sse \Sigma$ be a set of symbols. We define the set of codewords that are \emph{$t$-queried by $S$} to be
    $$
    \mathcal{C}[S,t]=\{c\in\mathcal{C}: |\{ c_1,\ldots,c_n\} \cap S|\ge t\}.
    $$
\end{definition}

\begin{definition}[List Recoverability]
\label{def:list recoverability}
We say that a code $\mathcal{C}\sse \Sigma_1\times\cdots\times \Sigma_n$ is $(\ell, L,t)$-list recoverable if for any set $S\sse \Sigma$ of symbols of size at most $\ell$, we have
    $|\mcC[S,t]|\leq L.$
\end{definition}

The standard notion of list-recoverability bounds the number of symbols in each block by $\ell$ rather than the total number of symbols across all blocks by $\ell$. \Cref{def:list recoverability} is slightly more convenient for our purposes. These definition are equivalent up to a factor of $n$, which will be negligible for the parameter regimes that we work in.

\pparagraph{Properties of the code that \cite{YZ22} uses.}
We will use the same code construction as \cite{YZ22}, and therefore, their quantum upper bound will apply as-is to our hard problem. These codes satisfy the following properties which we use.
\begin{fact}[Suitable codes; Lemma~4.2 of \cite{YZ22}]\label{fact:code-exists}
    For any constants $0<c<c'<1$, there exists an explicit linear code $\mathcal{C}\sse \Sigma_1\times\cdots\times \Sigma_n$ such that $|\Sigma|= 2^{\Theta(n\log n)}$ and $|\mcC|\geq 2^{a n}$ for some constant $a>1$. Furthermore, 
    \begin{enumerate}
        \item \label{fact:code-exists--item:list-recoverable}$\mcC$ is $(\ell, L, t)$-list recoverable where $t=(1-\zeta)n$ for some constant $0<\zeta<1$, $\ell=2^{\Theta(n^c)}$, and $L=2^{O (n^{c'})}$; and
        \item  \label{fact:code-exists--item:hw} For all $l\in[n-1]$,
        \begin{equation*}
            \Prx_{\bc \sim \Unif(\mcC)}[\hw(\bc)=n-l]\leq \left(\frac{n}{|\Sigma|}\right)^l.
        \end{equation*}
        where $\hw(\bc)$ denotes the Hamming weight of $\bc$.
    \end{enumerate}
\end{fact} \cite{YZ22} apply list recoverability bounds of Guruswami and Rudra~\cite{GR08, Rud07} to prove the existence of codes satisfying the above criteria. Using additional properties of the code $\mcC$,~\cite{YZ22} construct a quantum sampler for $\mcC_{\bmcO}$.
\begin{theorem}[Quantum easiness of sampling $\Unif(\mcC_{\bmcO})$]
\label{thm:YZ quantum easiness}
    For any $n \in \N$, there exists a $\poly(n)$-time quantum algorithm that, given oracle access to a random oracle $\bmcO$, with high probability over $\bmcO$, generates samples from a distribution $\mcD$ satisfying $\dtv(\mcD,\Unif(\mcC_{\bmcO})) \leq 2^{-\Omega(n)}$.
\end{theorem}


\paragraph{Decision trees and forests.}
\label{par:trees-forests}
We consider decision trees (i.e.~query algorithms) and forests (i.e.~families of query algorithms) with internal nodes labeled by symbols in $\Sigma$ and leaves that output an element of $\Sigma_1 \times \cdots \Sigma_n$. For a decision tree $T$ and oracle $\mcO : \Sigma \to \zo$ we write $T(\mathcal{O})$ to the output of $T$ when branches are taken according to $\mcO$.  We write $\Path(T,\mcO) \sse \Sigma$ to be the set of symbols that $T$ queries on the path determined by $\mcO$.

We will assume, without loss of generality, that all trees are \emph{valid}: 
\begin{definition}[Valid trees]
    \label{def:valid-tree}
    A tree $T$ is \emph{valid} if, on any oracle $\mcO$, the output $c = T(\mcO)$ is fully queried, meaning $c_i \in \Path(T,\mcO)$ for all $i \in [n]$.
\end{definition}
The reason we may assume validity without loss of generality is that any $T'$ that is not valid can be easily converted to an equivalent valid tree $T$: Specifically, $T(\mcO)$ first simulates $c = T'(\mcO)$, then queries $c_1,...,c_n$, and finally outputs $c$. The result is that $T$ has equivalent input/output behavior to $T$, is valid, and makes at most $n$ additional queries. Since our lower bounds are against trees making well more than $n$ queries, this additive factor of $n$ only affects the constant in our asymptotic statements.


\paragraph{The connection between oracle circuits and trees}
\begin{definition}[Oracle circuit]
    A size-$s$ \emph{oracle circuit} $G$, is a generalization of a size-$s$ circuit that is allowed to have \textsc{Oracle} gates in addition to the standard \textsc{And}, \textsc{Or}, and \textsc{Not} gates. Given an oracle $\mcO$ and input $x$, the output $G^{\mcO}(x)$ is computed by using $x$ as the input to $G$ and replacing all \textsc{Oracle} gates with calls to $\mcO$.
\end{definition}
We will use a standard connection between oracle circuits and query algorithms.
\begin{fact}
    \label{fact:connect-oracle-tree}
    For any size-$s$ oracle circuit $G$ and input $x$, there exists a depth-$s$ tree $T_x$ such that $G^{\mcO}(x) = T_x(\mcO)$ for every oracle $\mcO$.
\end{fact}
The tree $T$ in \Cref{fact:connect-oracle-tree} is constructed by replacing every oracle call of $G$ with a query in $T$.

\section{Proof of \Cref{lem:hardness amplification intro}: Hardness amplification for {\sc NullCodeword}}
\label{sec:hardness amplification}

{Recall the following definitions from the introduction. In the search problem {\sc NullCodeword}, the goal is to find an element of the set $\mcC_{\mcO}$:

\begin{definition}[The set $\mcC_{\mcO}$]
    Let $\mathcal{C}\sse \Sigma_1 \times \cdots \times \Sigma_n$ be a code and $\mcO : \Sigma \to \zo$ be an oracle.  We define the set $\mcC_{\mcO}\sse \mcC$, 
    \[ \mcC_{\mcO} \coloneqq \{ c \in \mcC \colon \mcO(c) = 0^n\},  \]
    where $\mcO(c) \coloneqq (\mcO(c_1),\ldots,\mcO(c_n))$. 
\end{definition}

In this section, we prove our hardness amplification lemma for {\sc NullCodeword}. This lemma says that it is hard for a forest to find many distinct elements of $\mcC_{\mcO}$:

\begin{definition}[$k$-succeeds]
\label{def:k-succeeds}
Let $\mcC$ be a code and $\mcO$ be an oracle. Let $\mathcal{F}=\{T^{(1)},\ldots, T^{(m)}\}$ be a forest. We say that $\mcF$ \emph{$k$-succeeds on $\mcO$} if $\mathcal{F}(\mcO) = \{T^{(1)}(\mcO),\ldots, T^{(m)}(\mcO)\}$ outputs at least $k$ unique elements of $\mathcal{C}_\mcO$.
\end{definition}


\begin{lemma}[Restatement of \Cref{lem:hardness amplification intro}]
\label{lem:forest fails}
Let $\mcC$ be an $(\ell,L,t)$-list recoverable code and $\mcF$ be a size-$m$ forest of depth $q$. 
Then 
\[    \Prx_{\bmcO}[\mathcal{F}\text{ }(L+1)\text{-succeeds on }\bmcO ]\le (mL\cdot2^{-(n-t)})^{\floor{\ell/q}}.
    \]
\end{lemma}

\subsection{Hardness of $S$-conditioned {\sc NullCodeword}}
We begin with a helper lemma for our proof of~\Cref{lem:forest fails}. 

\begin{lemma}
\label{lem:main-lemma}
Let $\mcC$ be an $(\ell,L,t)$-list recoverable code. Let $T$ be a depth-$q$ tree. Let $S \sse \Sigma$ be a set of $\leq \ell - q$ symbols and $\mcO_S : S \to \zo$ be a partial oracle over $S$.  We have
    $$
    \Prx_{\bmcO}\big[T(\bmcO)\in \mcC_{\bmcO}\setminus\mathcal{C}[S,t]\mid \bmcO_S \equiv \mcO_S\big]\le L \cdot 2^{-(n-t)}.
    $$
\end{lemma}

The $S = \emptyset$ special case of~\Cref{lem:main-lemma} recovers~\cite{YZ22}'s classical query lower bound for {\sc NullCodeword} (Lemma 6.9 of~\cite{YZ22}).\footnote{\cite{YZ22} proved their lower bound with $\ell, L,$ and $t$ instantiated as specific functions of $n$, but an inspection of their proof shows that it establishes the more general statement corresponding to~\Cref{lem:main-lemma} with $S = \emptyset$.} \Cref{lem:main-lemma} says that {\sc NullCodeword} remains hard when conditioned on any small set $S$ of values of the oracle---we think of these values as being revealed to the tree for free---and the tree is asked to find a new codeword in $\mcC_{\mcO}$ that is sufficiently disjoint from $S$ in the sense of lying outside $\mcC[S,t]$.

\paragraph{A key event used in the proof of \Cref{lem:main-lemma}.} For any codeword $c \in \mcC \setminus \mcC[S, t]$,
 we define the event: 
\[ Q(c,\mcO) \coloneqq \Ind[\text{$c$ is $t$-queried by $S\cup \Path(T,\mcO)$}]. \]
We drop the dependence on $S$, $T$, and $t$ since they are fixed in the statement of~\Cref{lem:main-lemma}. 
Before proving the main lemma, we prove a few supporting claims:

\begin{claim}
\label{clm:event-small-list}
Regardless of $\mcO$, the number of $c$ for which $Q(c,\mcO)$ holds is at most $L$.
\end{claim}

\begin{proof}
Observe that $|S \cup \Path(T, \mcO)| \le \ell$.
This is because $T$ is a $q$-query algorithm; thus $|\Path(T, \mcO)| \le q$, and by the assumption of \Cref{lem:main-lemma}, $|S| \le \ell - q$.
Then by $(\ell, L, t)$-list-recoverability of $\mathcal{C}$, 
it follows that the number of codewords that are
$t$-queried by $S \cup \Path(T,\mcO)$ is at most $L$.
\end{proof}

\begin{proposition}
    \label{prop:disjoint-conditioning}
    Let $\bE_1,\ldots, \bE_m$ be any \emph{disjoint} events, and $\bF$ be any other event. For $\bE \coloneqq \bE_1 \cup \cdots \cup \bE_m$,
    \begin{equation*}
        \Pr[\bF \mid \bE] \leq \max_{i \in [m]} \Pr[\bF \mid \bE_i].
    \end{equation*}
\end{proposition}
\begin{proof}
    Let $\bz$ be a random variable taking on $\varnothing$ if $\bE$ does not occur, and otherwise taking on the unique choice of $i$ for which $\bE_{i}$ occurs. Then,
    \begin{align*}
        \Pr[\bF \mid \bE] &= \Pr[\bF \mid \bz \in [m]] \\
        &= \Ex_{\bz}\bracket*{\Pr[\bF \mid \bz] \mid \bz \in [m]} \\
        &\leq \max_{z \in [m]} \Pr[\bF \mid \bz = z] \leq \max_{i \in [m]} \Pr[\bF \mid \bE_i].\qedhere
    \end{align*}
\end{proof}
\begin{claim}
\label{clm:many-free-bits}
    For any $c \in \mcC \setminus \mcC[S, t]$, 

\[\Prx_{\bmcO}[c \in \mcC_{\bmcO} \mid Q(c, \bmcO)\text{ and }\bmcO_S \equiv \mcO_S] \le 2^{-(n-t)}.\]
\end{claim}
In this proof we will use the notation $\Path(T, v)$ to denote the set of symbols queried in $T$ in the path terminating at node $v$.
\begin{proof}
By the assumption that $c \not\in \mcC[S,t]$,
it must be the case that $c$ is at most $(t-1)$-queried by $S$.
Consider the set of nodes of $T$ that make the $t_{th}$ query to $c$ (by the assumption that $Q(c, \bmcO)$ occurs and $S$ at most $(t-1)$-queries $c$, this set must not be empty).
In other words,
this set contains the nodes $v$ such that $(\Path(T, v) \setminus \{v\}) \cup S$
contains exactly $t$ symbols of $c$.
The event $Q(c, \bmcO)$ is exactly the event that $\bmcO$ reaches such a node. 
We will denote this set of nodes $V$.

We condition on $Q(c, \bmcO)$ and $\bmcO_S \equiv \mcO_S$.
Since the events $\{\Ind[\bmcO$ reaches $v]\mid v \in V\}$ are disjoint, we have by \Cref{prop:disjoint-conditioning}:
\[\Prx_{\bmcO}[c \in \mcC_{\bmcO} \mid Q(c, \bmcO)\text{ and }\bmcO_S \equiv \mcO_S]  \le \max_{v \in V}\big\{ \Prx[c \in \mcC_{\bmcO}\mid \bmcO\text{ reaches $v$ and }\bmcO_S \equiv \mcO_S] \big\}.\]
Since $(\Path(T, v) \setminus \{v\})\cup S$ contains $t$ symbols of $c$,
conditioning on this event leaves $n-t$ symbols of $c$ unconditioned.
These $n-t$ symbols are independent of all other symbols of $\bmcO$ (because $\bmcO$ is a fully random oracle). For the event $c \in \mcC_{\bmcO}$ to occur,
all of the $n-t$ unconditioned symbols must evaluate to $0$, which occurs with probability $2^{-(n-t)}$.
\end{proof}

\subsubsection{Putting the claims together to prove~\Cref{lem:main-lemma}}
Now we will prove the main lemma of this section, using the claims we have shown.
\begin{proof}[Proof of \Cref{lem:main-lemma}]

We aim to upper bound 
\[\Prx_{\bmcO}\big[T(\bmcO)\in \mcC_{\bmcO}\setminus\mathcal{C}[S,t]\mid \bmcO_S \equiv \mcO_S\big],\]
which is the probability that $T$ produces a new correct output (i.e.~one in $\mcC_{\mcO}$) that hasn't yet been $t$-queried by the set $S$.

We have:
\[\Prx_{\bmcO}\big[T(\bmcO)\in \mcC_{\bmcO}\setminus\mathcal{C}[S,t]\mid \bmcO_S \equiv \mcO_S\big] = \sum_{c \in  \mcC \setminus \mcC[S, t]} \Prx_{\bmcO}\bracket[\big]{T(\bmcO) = c \text{ and }c \in \mcC_{\bmcO} \mid \bmcO_S \equiv \mcO_S}.\]
Recall that we require every tree to $n$-query its output (see \Cref{par:trees-forests} for discussion on this.)
Therefore, if $T$ outputs $c$, then it must be the case that $Q(c, \bmcO)$ occurs.
So we have:
\begin{align*}
 \sum_{c \in  \mcC \setminus \mcC[S, t]} &\Prx_{\bmcO}\bracket[\big]{T(\bmcO) = c \text{ and }c \in \mcC_{\bmcO} \mid \bmcO_S \equiv \mcO_S} \\
    &= \sum_{c \in  \mcC \setminus \mcC[S, t]}  \Prx_{\bmcO}[Q(c, \bmcO) \mid \bmcO_S \equiv \mcO_S]\cdot \Prx_{\bmcO}\bracket[\big]{T(\bmcO) = c \text{ and }c \in \mcC_{\bmcO} \mid Q(c, \bmcO) \text{ and } \bmcO_S \equiv \mcO_S} \\
    &\leq \sum_{c \in  \mcC \setminus \mcC[S, t]} \Prx_{\bmcO}[Q(c, \bmcO) \mid \bmcO_S \equiv \mcO_S] \cdot \Prx_{\bmcO}\bracket[\big]{c \in \mcC_{\bmcO} \mid Q(c, \bmcO) \text{ and }\bmcO_S \equiv \mcO_S}.
\end{align*}
Now we apply the bounds of \Cref{clm:many-free-bits} and \Cref{clm:event-small-list}, which concludes the proof:
\begin{align*}
    &\le \sum_{c \in  \mcC \setminus \mcC[S, t]} 2^{-(n-t)} \cdot \Prx_{\bmcO}[Q(c, \bmcO) \mid \bmcO_S \equiv \mcO_S]\tag{\Cref{clm:many-free-bits}}\\ 
    &= 2^{-(n-t)} \cdot \Ex_{\bmcO}\bracket[\bigg]{ \sum_{c \in  \mcC \setminus \mcC[S, t]} Q(c, \bmcO) \,\bigg|\, \bmcO_S \equiv \mcO_S} \tag{Linearity of expectation}\\
    &\leq 2^{-(n-t)} \cdot \max_{\mcO \text{ consistent with }\mcO_S}\set[\Bigg]{ \sum_{c \in  \mcC \setminus \mcC[S, t]} Q(c, \mcO)} \tag{Expectation at most max}\\
    &\le L \cdot 2^{-(n-t)}\tag{\Cref{clm:event-small-list}}.
\end{align*}
\end{proof}

\subsection{Proof of~\Cref{lem:forest fails} using~\Cref{lem:main-lemma} }


We now introduce a crucial definition, \emph{well-spread} sequences of trees. This definition requires that the output of one tree is (mostly) disjoint of all variables queried by previous trees.
\begin{definition}[Well-spread sequence of trees]
    \label{def:well-spread}
    We say a sequence of trees $T^{(1)}, \ldots, T^{(r)}$ is \emph{$t$-well-spread} on an oracle $\mcO$, if, using $S_i \coloneqq \bigcup_{j=1}^i \Path(T^{(i)}, \mcO)$ to denote the variables queried by $T^{(1)}, \ldots, T^{(i)}$ and $c_{i+1} \coloneqq T^{(i+1)}(\mcO)$ the output of the $(i+1)^{\text{st}}$ tree,
    \begin{equation*}
        c_{i+1} \text{ is not }t\text{-queried by }S_i.
    \end{equation*}
\end{definition}



To see why \Cref{def:well-spread} is useful, we can apply \Cref{lem:main-lemma} to show that it is very unlikely for a sequence of trees to be unpredictable and still output elements of $\mcC_{\bmcO}$.
\begin{claim}[Well-spread sequences rarely succeed]
\label{claim:unlikely well spread and all zeroes}
    Let $\mcC$ be an $(\ell, L, t)$-list recoverable code and $T^{(1)}, \ldots, T^{(r)}$ be a sequence of valid (as in \Cref{def:valid-tree}) depth-$q$ trees where $qr \leq \ell$. The probability over a uniform $\bmcO$ that $T^{(1)}, \ldots, T^{(r)}$ are $t$-well-spread and $r$-succeed on $\bmcO$ is at most $(L \cdot 2^{-(n-t)})^r$.
\end{claim}
We will also show that we can always find a well-spread sequence.
\begin{claim}[Every successful forest has a well-spread sequence]
    \label{clm:well-spread-exists}
    Let $\mcC$ be a $(\ell, L, t)$-list-recoverable code. For any depth-$q$ forest $\mcF \coloneqq T^{(1)}, \ldots, T^{(m)}$ which $L+1$-succeeds on $\mcO$ and $r$ satisfying $(r-1)q \leq \ell$, there exists $i_1,\ldots, i_r \in [m]$ for which $T^{(i_1)},\ldots, T^{(i_r)}$ is $t$-well-spread and $r$-succeeds on $\mcO$.
\end{claim}

Before proving \Cref{claim:unlikely well spread and all zeroes,clm:well-spread-exists}, we show that together they easily imply \Cref{lem:forest fails}.
\begin{proof}[Proof of \Cref{lem:forest fails} assuming \Cref{claim:unlikely well spread and all zeroes,clm:well-spread-exists}]
    Let $r \coloneqq \floor{\ell/q}$ and $\mcF \coloneqq \set{T^{(1)}, \ldots, T^{(m)}}$. Then, by \Cref{claim:unlikely well spread and all zeroes}, for any fixed choice of $i_1,\ldots, i_r \in [m]$, the probability that $T^{(i_1)}, \ldots T^{(i_r)}$ are $t$-well-spread and $r$-succeed on $\bmcO$ is at most $(L \cdot 2^{-(n-t)})^r$. There are at most $m^r$ choices for $i_1,\ldots, i_r$. Therefore, by union bound, the probability there exists any $i_1,\ldots, i_r \in [m]$ for which $T^{(i_1)}, \ldots T^{(i_r)}$ are $t$-well-spread and $r$-succeed on $\bmcO$ is at most $m^r \cdot (L \cdot 2^{-(n-t)})^r$.

    By \Cref{clm:well-spread-exists}, in order $\mcF$ to $(L+1)$-succeed on $\bmcO$, there must be such a choice of $i_1,\ldots, i_r \in [m]$. Therefore,
    \begin{equation*}
        \Prx_{\bmcO}[\mathcal{F}\text{ }(L+1)\text{-succeeds on }\bmcO ]\le (mL\cdot2^{-(n-t)})^r = (mL\cdot2^{-(n-t)})^{\floor{\ell/q}}. \qedhere
    \end{equation*}
\end{proof}

\subsubsection{Proof of \Cref{claim:unlikely well spread and all zeroes}}
\begin{proof}
    Let $E_{\bmcO}(T^{(1)}, \ldots, T^{(r)})$ be the event indicating that $T^{(1)}, \ldots, T^{(r)}$ are $t$-well-spread and $r$-succeed on $\bmcO$. We will show that for any $i < r$,
    \begin{equation}
        \label{eq:succeed-conditional}
        \Pr[E_{\bmcO}(T^{(1)}, \ldots, T^{(i+1)}) \mid E_{\bmcO}(T^{(1)}, \ldots, T^{(i)})] \leq L \cdot 2^{-(n-t)},
    \end{equation}
    which easily implies the desired result.

    Let $\bS_i \coloneqq \bigcup_{j=1}^i \Path(T^{(i)}, \bmcO)$ be the variables of $\bmcO$ queried by $T^{(1)}, \ldots, T^{(i)}$. The outputs of $T^{(1)}, \ldots, T^{(i)}$ are entirely determined by $\bmcO_{\bS_i}$. Furthermore, since the trees are valid (meaning each tree entirely queries its output), whether $T^{(1)}, \ldots, T^{(i)}$ $i$-succeed is also entirely determined by $\bmcO_{\bS_i}$. Therefore, the event $E_{\bmcO}(T^{(1)}, \ldots, T^{(i)})$ is entirely determined by $\bmcO_{\bS_i}$. Applying \Cref{prop:disjoint-conditioning}, we have that
     \begin{equation*}
         \Pr[E_{\bmcO}(T^{(1)}, \ldots, T^{(i+1)}) \mid E_{\bmcO}(T^{(1)}, \ldots, T^{(i)})] \leq \max_{S_i, \mcO_{S_i}} \Pr[E_{\bmcO}(T^{(1)}, \ldots, T^{(i+1)}) \mid \bmcO_{\bS_i} = \mcO_{S_i}],
     \end{equation*}
     where the maximum is taken over all choices of $S_i$ and $\mcO_{S_i}$ that result in $ E_{\bmcO}(T^{(1)}, \ldots, T^{(i)})$ occurring. Next, we observe that in order for $E_{\bmcO}(T^{(1)}, \ldots, T^{(i+1)})$, it must be the case that
    \begin{equation*}
        T^{(i+1)} \in \mcC_{\bmcO} \setminus \mcC[\bS_i, t].
    \end{equation*}
    Therefore, we have that
    \begin{equation*}
        \Pr[E_{\bmcO}(T^{(1)}, \ldots, T^{(i+1)}) \mid E_{\bmcO}(T^{(1)}, \ldots, T^{(i)})] \leq \max_{S_i, \mcO_{S_i}} \Pr[T^{(i+1)}(\bmcO) \in \mcC_{\bmcO} \setminus \mcC[\bS_i, t] \mid \bmcO_{\bS_i} = \mcO_{S_i}].
    \end{equation*}
    Applying \Cref{lem:main-lemma} and the fact that we only take the maximum over $|S_i|$ of size at most $q \cdot (r-1)$ recovers \Cref{eq:succeed-conditional}.
\end{proof}



\subsubsection{Proof of \Cref{clm:well-spread-exists}}

Since $\mcF$ $(L+1)$-succeeds on $\mcO$, there exists some set of $(L+1)$ coordinates $I \subseteq [m]$ for which all of $\set{T^{(i)}}_{i \in I}$ output distinct elements of $\mcC_{\mcO}$. We will use a simple greedy algorithm to choose $i_1, \ldots, i_r$ from this set $I$.

For each $j \in [r]$, we set $i_j$ to be any element of $I$ satisfying,
\begin{equation}
    \label{eq:choose-well-spread}
    T^{(i_j)}(\mcO)\text{ is not $t$-queried by } \bigcup_{a = 1}^{j-1}\Path(T^{(i_a)}, \mcO).
\end{equation}
By definition, the resulting $T^{(i_1)},\ldots, T^{(i_r)}$ are $t$-well-spread. Furthermore, they will all output unique elements of $\mcC_O$, since every index we pick is in the set $I$, so they $r$-succeed. All that remains is to justify that this greedy procedure never gets stuck; i.e., there is always a choice for $i_j$ satisfying \Cref{eq:choose-well-spread}.

Each of the $L+1$ trees $\set{T^{(i)}}_{i \in I}$ output a unique element of $\mcC$. The total number of variables contained in $S_{j-1} \coloneqq \bigcup_{a = 1}^{j-1}\Path(T^{(i_a)}, \mcO)$ is at most $(r-1) \cdot q$, which is at most $\ell$ (by assumption of \Cref{clm:well-spread-exists})  . Therefore, the number of elements of $\mcC$ that are $t$-queried by $S_{j-1}$ is at most $L$. Therefore, there must exist at least one choice of $i_j$ satisfying \Cref{eq:choose-well-spread}.

\section{Proof of~\Cref{thm:main intro} using~\Cref{lem:forest fails}}

In this section we prove: 


\begin{theorem}
    \label{thm:classical-lb-hitting}
      Let $\mcC$ be the code from \Cref{fact:code-exists}. With probability at least $1 - 2^{-\Omega(n)}$ over a random oracle $\bmcO:\Sigma \to \zo$, there is no size $s \coloneqq 2^{n^{O(1)}}$ oracle circuit $G$ such that $G^{\bmcO}$ samples a distribution that is $(1 - 2^{-\Omega(n)})$-close in TV distance to $\Unif(\mcC_{\bmcO})$.
\end{theorem}

\paragraph{Proof of~\Cref{thm:main intro} from \Cref{thm:classical-lb-hitting}.} We first must move from distributions over $\Sigma^n$ to distributions over $\zo^N$. This is straightforward: For $r = \ceil{\log |\Sigma|} = \tilde{O}(n)$, there are efficient mappings $f:\Sigma \to \zo^r$ and $g:\zo^r \to \Sigma$ satisfying that $g(f(\sigma)) = \sigma$ for all $\sigma \in \Sigma$. Then, for $N \coloneqq n \cdot r$ and oracle $\bmcO_{\mathrm{boolean}}:\zo^r \to \zo$, we define the distribution $\mcD^{(N)}_{\bmcO_{\mathrm{boolean}}}$ to be the distribution formed by 
\begin{enumerate}
    \item  Defining the oracle $\bmcO:\Sigma \to \zo$ as $\bmcO \coloneqq \bmcO_{\mathrm{boolean}} \circ f$.
    \item Sampling a code word $\bc \sim \Unif(\mcC_{\bmcO})$.
    \item Outputting the $N$ bit string $(f(\bc_1),\ldots, f(\bc_n))$.
\end{enumerate}
Our lower bound from \Cref{thm:classical-lb-hitting} extends to a lower bound for $\mcD^{(N)}_{\bmcO_{\mathrm{boolean}}}$ because if a circuit could efficiently sample from $\mcD^{(N)}_{\bmcO_{\mathrm{boolean}}}$, it could apply $g$ to efficiently sample from $\Unif(\mcC_{\bmcO})$, which \Cref{thm:classical-lb-hitting} rules out. For the same reason, \cite{YZ22} (\Cref{thm:YZ quantum easiness}) construct an efficient quantum algorithm which samples a distribution that is $2^{-\Omega(n)}$-close to $\Unif(\mcC_{\bmcO})$ and this can easily be transformed to an efficient quantum sampler for  a distribution $2^{-\Omega(n)} = 2^{-\tilde{\Omega}(\sqrt{N})}$-close to $\mcD^{(N)}_{\bmcO_{\mathrm{boolean}}}$ by applying $f$.

Summarizing, we have constructed a distribution over $N$ bits that has an $2^{-\tilde{\Omega}(\sqrt{N})}$-approximate efficient quantum sampler, but no $(1 - 2^{-\tilde{\Omega}(\sqrt{N})})$-approximate efficient classical sampler. Taking $\wh{\mcD}^{(N)}_{\bmcO}$ to be the distribution this quantum algorithm outputs, we directly have a quantum sampler for exactly sampling $\wh{\mcD}^{(N)}_{\bmcO}$. Furthermore, by the triangle inequality for TV distance, no efficient classical sampler can $(1 - 2\cdot 2^{-\tilde{\Omega}(\sqrt{N})})$-approximately sample $\wh{\mcD}^{(N)}_{\bmcO}$. This completes the proof of~\Cref{thm:main intro}, with the hard distribution being $\wh{\mcD}^{(N)}_{\bmcO}$.

\subsection{A helper claim: From oracle circuits to collections of forests}

We first prove a helper claim (\Cref{claim:national park}). This claim shows how we can extract from an oracle circuit $G$, a small collection of forests with the following property: If $G^{\mcO}$ samples a distribution that is even non-negligibly close to $\mathrm{Unif}(\mcC_{\mcO})$, then one of the forests in this collection will $k$-succeed. In the next subsection we use~\Cref{claim:national park} and~\Cref{lem:main-lemma} to prove~\Cref{thm:main intro}. 

We begin with a probabilistic fact that will be useful for our proof of~\Cref{claim:national park}:

\begin{claim}[Pairwise independent draws contain many unique elements]
    \label{claim:pairwise-hits-often}
    Let $m,\eps>0$ and $H\subseteq\zo^n$ such that $|H|\geq m$. Let $\mcD$ be any distribution that is $(1-\eps)$-close in TV distance to $\Unif(H)$. Then, if $\bx^{(1)}, \ldots, \bx^{(m)}$ are pairwise independent samples each marginally from $\mcD$, with nonzero probability, they will contain at least $k= \frac{m\eps}{2}$ unique elements of $H$.
\end{claim}

Our bound on $k$ is optimal in expectation up to the factor of $2$: 
If we set $\mcD$ to be a mixture distribution which is uniform over $|H|$ with probability $\eps$ and otherwise disjoint from $|H|$, then by taking $m$ samples from $\mcD$, the expected number of element of $H$ sampled is at most~$m \eps$.

\begin{proof} 
 We begin by calculating the probability that any specific $x$ appears in $\bx^{(1)}, \ldots, \bx^{(m)}$. 
 For each $i \in [m]$ let $\bz_i = \Ind[\bx_i = x]$ and $\bZ = \sum_{i \in [m]} \bz_i$ be the number of times $x$ appears in the sample. Then, the $\bz_i$ are pairwise independent and each have expectation $\mcD(x)$. Therefore, denoting $p \coloneqq \mcD(x)$,
    \begin{equation*}
        \Ex[\bZ] = pm \quad\quad\text{and}\quad\quad\Var[\bZ] = mp(1-p)\leq pm.
    \end{equation*}
    Applying the second moment method,
    \begin{align*}
         \Pr\bracket[\Big]{x \in \set[\Big]{\bx^{(1)}, \ldots, \bx^{(m)}}} &= \Pr[\bZ > 0] \\
         &\geq \frac{\Ex[\bZ]^2}{\Ex[\bZ^2]} \\
         &\geq \frac{p^2m^2}{p^2m^2 + pm} \\
         &= \frac{pm}{pm + 1}\\
         & \geq \frac12 \min\set*{1, pm}.
    \end{align*}

Let $\bU =\sum_{x\in H} \Ind\bracket[\Big]{x \in \set[\Big]{\bx^{(1)}, \ldots, \bx^{(m)}}} $ be the total number of unique elements in $H$. Then, 
\begin{align*}
    \Ex[\bU] &= \sum_{x\in H} \Pr\bracket[\Big]{x \in \set[\Big]{\bx^{(1)}, \ldots, \bx^{(m)}}}\\
    &\geq\frac12 \sum_{x\in H}  \min\set*{1,m\mcD(x)}\\
    &\geq \frac12 \sum_{x\in H}  \min\set*{\frac{m}{|H|},m\mcD(x)}. &\tag{$|H|\geq m$}
\end{align*}

Finally, we lower bound the probability mass that $\mcD(x)$ must place on $H$ using our constraint on TV distance. 
\begin{equation*}
    \sum_{x\in H} \left(\frac{1}{|H|} -  \min\set*{\frac{1}{|H|}, \mcD(x)} \right) =\sum_{x\in H | \mcD(x)\leq 1/|H|} \left(\frac{1}{|H|} -  \mcD(x)\right)\\
    \leq \dist_{TV}(\Unif(H),\mcD)\tag{definition of $\dist_{TV}$}\\
    \leq 1-\eps.
\end{equation*} 
Rearranging the above equation yields $\sum_{x\in H} \min\set*{\frac{1}{|H|}, \mcD(x)} \geq \eps$, or equivalently, 
\begin{equation*}
    \sum_{x\in H} \min\set*{\frac{m}{|H|}, m\mcD(x)} \geq m\eps. 
\end{equation*}

Combining this bound with our lower bound on $\Ex[\bU]$ gives
\begin{equation*}
     \Ex[\bU] \geq \frac{m\eps}{2}.
\end{equation*}
 Since this is the expected number of distinct elements, we certainly achieve at least this many with nonzero probability. \end{proof}

\begin{fact}[Randomness-efficient generation of $d$-wise independent random variables~\cite{Jof74}]
\label{fact:explicit prgs for rvs}
    For all $d$, $n$, and $m\le 2^n$, there exists a circuit $P:\zo^\lambda\to (\zo^n)^m$  s.t. $P(\Unif)$ generates a collection of $m$-many marginally uniform $d$-wise independent random variables over $\zo^n$ with $\lambda=O(dn)$. 
\end{fact}

We are now ready to prove the helper claim:

\begin{claim}[Each oracle circuit $\to$ Collection of forests]
\label{claim:national park}
    Let $m\leq 2^n$. For every size-$s$ oracle circuit $G$, there exists a collection of $2^{O(s)}$ forests $\mcF^{(1)}, \ldots, \mcF^{(2^{O(s)})}$, each containing $m$ depth-$s$ trees, with the following property: For any oracle $\mcO$, if $|\mcC_{\mcO}|\geq m$ and $G^{\mcO}$ samples a distribution $(1-\eps)$-close in TV distance to $\Unif(\mcC_{\mcO})$, then there is some $i \in [2^{(O(s)}]$ for which $\mcF^{(i)}$ $k$-succeeds on $\mcO$, where $k\geq \frac{m\eps}{2}$.
\end{claim} 

\begin{proof}
Let $G: \zo^{\lambda} \to \zo^n$ be a size-$s$ oracle circuit, generating a distribution $\mcD$.  We observe that if $\ba^{(1)}, \ldots, \ba^{(m)}$ are pairwise independent and each uniform on $\zo^\lambda$, then \[ \bx^{(1)} = G(\ba^{(1)}), \ldots, \bx^{(m)} = G(\ba^{(m)})\] are pairwise independent and each have $\mcD$ as their marginal distribution. We use \Cref{fact:explicit prgs for rvs} to generate such $\ba^{(1)}, \ldots, \ba^{(m)}$. Specifically, there is a pairwise independent generator $P:\zo^{O(\lambda)} \to (\zo^\lambda)^m$ so that for $\br \sim \zo^{O(\lambda)}$ drawn uniformly, $P(\br)$ produces $\ba^{(1)}, \ldots, \ba^{(m)}$ that are pairwise independent and each uniform on $\zo^\lambda$. Let $P(\br)^{(i)}$ denote the corresponding $\ba^{(i)}$. In this notation, the random variables  \[ \bx^{(1)} =G(P(\br)^{(1)}),\dots, \bx^{(m)} =G(P(\br)^{(m)}) \] are pairwise independent with $\mcD$ as its marginal. 

We now describe how to construct the collection of forests from $G$. There are $2^{O(\lambda)}\leq 2^{O(s)}$ many seeds $r$, and each gives rise to a different forest. For a specific seed $r$, the associated forest is the trees $T^{(1)},\dots, T^{(m)}$ where the tree $T^{(i)}$ outputs $G(P(r)^{(i)})$. By \Cref{fact:connect-oracle-tree}, each such tree $T^{(i)}$ exists and has depth $s$.


If $|\mcC_{\mcO}|\geq m$ and the distribution $\mcD$ that $G^{\mcO}$ samples is  $(1-\eps)$-close to $\Unif(\mcC_{\mcO})$, then the random variables $\bx^{(1)} =G(P(\br)^{(1)}),\dots, \bx^{(m)} =G(P(\br)^{(m)}$ meet the conditions of \Cref{claim:pairwise-hits-often}. Therefore, with nonzero probability, it will contain $k\geq \frac{m\eps}{2}$ unique elements of $\mcC_{\mcO}$. In other words, there must be a single choice of the seed $r \in \zo^{O(\lambda)}$ for which the corresponding forest $T^{(1)},\dots, T^{(m)}$ $k$-succeeds.  
\end{proof}

\subsection{Proof of~\Cref{thm:classical-lb-hitting} from~\Cref{lem:forest fails} and~\Cref{claim:national park}}

 First, we will prove that the set $\mcC_{\bmcO}$ has exponential size with all but negligible probability. Doing so will allow us to meet the conditions of~\Cref{claim:national park}. 

\begin{claim}[$\mcC_{\bmcO}$ is large w.h.p.]
    \label{claim:set-is-big}
    With probability $1-2^{-\Omega(n)}$ over the random oracle $\bmcO$, we have
    $|\mcC_{\bmcO}|\geq 2^{\Omega(n)}$.   
\end{claim}
\begin{proof}

Let $\bz_i \coloneqq \Ind[\bmcO(c_i)=0^n]$ be the indicator that the $i$th codeword hashes to $0^n$, and let $d \coloneqq |\mcC|$. Then, $\bZ \coloneqq \sum_{i\in [d]} \bz_i$ equals $|\mcC_{\bmcO}|$ and 
\[ 
    \Ex_{\bmcO}[\bZ] =  \sum_{i\in [d]} \Ex_{\bmcO}\left[\bz_i\right] = d\cdot 2^{-n}
\]
To upper bound the second moment, first note that 
\[ \Ex_{\bmcO}\left[\bz_{\bi} \bz_{\bj}\right] = 2^{-(n + \hw(c_{\bi}-c_{\bj}))}\]
because in order for both $c_{\bi}$ and $c_{\bj}$ to hash to all 0, the symbols for which they differ must all hash to 0 (there are  $2 \cdot \hw(c_{\bi}-c_{\bj})$ such symbols), and the symbols on which they agree must also hash to 0 (there are $n - \hw(c_{\bi}-c_{\bj})$ such symbols). Therefore,
\begin{align*}
\Ex_{\bmcO}[\bZ^2] &=  \sum_{i\in [d]} \sum_{j\in [d]} \Ex_{\bmcO}\left[\bz_i \bz_j\right]\\
&= d^2 \Ex_{\bi, \bj \sim [d]} \left[ 2^{-(n + \hw(c_{\bi}-c_{\bj}))}\right]\\
 &=  d^2 \sum_{l\in[n]} 2^{-(n +n-l)} \Prx_{\bi,\bj\sim[d]}[\hw(c_{\bi}-c_{\bj})=n-l]\\
    &= d \cdot 2^{-n} +  d^2 \cdot 2^{-2n} \sum_{l\in [n-1]} 2^l \Prx_{\bi,\bj\sim[d]}[\hw(c_{\bi}-c_{\bj})=n-l]
    \end{align*}
    
    The last equation follows from considering $l=n$ separately. This value of $l$ corresponds to  the case when $\bi=\bj$, which happens with probability $\frac1d$. Applying \Cref{fact:code-exists--item:hw} of \Cref{fact:code-exists} to the difference $c_{\bi}-c_{\bj}$, we have 
\[ \Prx_{\bi,\bj \sim [d]} [ \hw(c_{\bi}-c_{\bj})= n-l]\leq \left(\frac{n}{|\Sigma|}\right)^l\] 
for all $l\in[n-1]$. Plugging this into the above equation yields

    \begin{align*}
    \Ex_{\bmcO}[\bZ^2]&\leq d \cdot 2^{-n} + d^2 \cdot 2^{-2n} \sum_{l\in [n-1]} 2^l \left(\frac{n}{|\Sigma|}\right)^l &\tag{\Cref{fact:code-exists}}\\
    &\leq d \cdot 2^{-n} + d^2 \cdot 2^{-2n} \left(\frac{1}{1-\frac{2n}{|\Sigma|}}\right) &\tag{Geometric series}\\
    &=  d \cdot 2^{-n} + d^2 \cdot 2^{-2n} \left(1+\frac{2n}{|\Sigma| -2n}\right).
\end{align*}
  Finally, we can calculate the variance:
\begin{align*}
\Var(\bZ)&= \Ex_{\bmcO}[\bZ^2]-\Ex_{\bmcO}[\bZ]^2\\
&\leq d \cdot 2^{-n} + d^2 \cdot 2^{-2n} \left(1+\frac{2n}{|\Sigma| -2n}\right) - d^2 \cdot 2^{-2n}\\
&= d \cdot 2^{-n} + d^2 \cdot 2^{-2n} \left(\frac{2n}{|\Sigma| -2n}\right).
\end{align*}
By Chebyshev's inequality, we have 
\begin{align*}
    \Pr \left[|\bZ -\Ex[\bZ]| \geq \frac{d\cdot 2^{-n}}{2} \right] 
    &\leq \frac{ 4 \Var[\bZ]}{d^2\cdot 2^{-2n}}\\
    &\leq 4\left( \frac{1}{d \cdot 2^{-n}} + \frac{2n}{|\Sigma| -2n}\right).
\end{align*}
By \Cref{fact:code-exists}, $d\geq 2^{a n}$ for some constant $a>1$, and $|\Sigma|= 2^{\Theta(n\log n)}$. Therefore, the right hand side of the above equation is exponentially small in $n$. Recalling  that $\Ex[\bZ]\geq d \cdot 2^{-n}$, we can conclude that $\bZ = |\mcC_{\bmcO}|\geq 2^{\Omega(n)}$ with all but negligible probability. 
\end{proof}

\newcommand{\bad}{\text{Bad}(G,\bmcO)}

\paragraph{Completing the proof of~\Cref{thm:classical-lb-hitting}.}

By \Cref{claim:set-is-big}, we are free to assume that $|\mcC_{\bmcO}|\geq 2^{bn}$ for some constant $b>0$.
Let $s=\sqrt{\ell}$; let $m=\min\set*{2^{bn},  2^{\zeta n/2}}$; and let $\eps = m^{-1/2}$ where $\ell$ and $\zeta$ are the list-recoverable parameters as in the statement of \Cref{fact:code-exists}. Let $\bad$ be the event that the oracle circuit $G$ samples  a distribution that is $(1-\eps)$-close to $\Unif(\mcC_{\bmcO})$. We will show that the probability of $\bad$ for any single size-$s$ oracle circuit  $G$ is much less than $s^{-\Omega(s)}$, so we can therefore afford to union bound over all size-$s$ circuits. 

Thus, we begin by fixing a single oracle circuit $G$. By~\Cref{claim:national park}, $G$ has an associated collection of $2^{O(s)}$ forests, and if $G$ samples a distribution close in TV distance to $\Unif(\mcC_{\bmcO})$, then one of these forests will $k$-succeed, where $k\geq \frac{m\eps}{2}$. Therefore, to bound the probability of $\bad$, it suffices to bound the probability that any of $G$'s associated forests $k$-succeeds. We first bound the probability of a single forest and then union bound over all $2^{O(s)}$ forests. 

By our choices of $m$ and $\eps$, we can conclude that $k\geq \frac{m\eps}{2}\geq 2^{\Omega(n)}$. By \Cref{fact:code-exists}, $L\leq2^{O(n^{c'})}\leq k$ (the inequality follows since $c'<1$), so for any forest $\mcF$ we have $ \Prx_{\bmcO}[\mathcal{F}\text{ }k\text{-succeeds on }\bmcO ]\leq  \Prx_{\bmcO}[\mathcal{F}\text{ }(L+1)\text{-succeeds on }\bmcO ]$. We can therefore apply \Cref{lem:forest fails} to calculate the probability that any fixed forest $\mcF$ $k$-succeeds: 
\begin{align*}
 \Prx_{\bmcO}[\mathcal{F}\text{ }k\text{-succeeds on }\bmcO ]&\le (mL\cdot2^{-(n-t)})^{\floor{\ell/q}}\\
 &\le (m\cdot 2^{O(n^{c'})} \cdot2^{-\zeta n})^{\floor{\ell/q}} &\tag{Setting $L$ and $t$ as in \Cref{fact:code-exists}}\\
 &\le (2^{O(n^{c'})} \cdot2^{-\zeta n/2})^{\floor{\ell/q}} &\tag{By choice of $m$}\\
 &\le (2^{-\Omega(n)})^{\floor{\ell/q}}&\tag{Since $c'<1$}\\
 &\le 2^{-\Omega(n\sqrt{\ell})},
\end{align*}
where the last inequality follows since the tree depth $q=s=\sqrt{\ell}$ as in the statement of \Cref{claim:national park}. Union bounding over all $2^{O(s)}=2^{O(\sqrt{\ell})}$ forests, we conclude that the probability that there exists a forest that $k$-succeeds is at most $2^{-\Omega(n\sqrt{\ell})} 2^{O(\sqrt{\ell})} = 2^{-\Omega(n\sqrt{\ell})}.$ Therefore, for a fixed circuit $G$, the probability of $\bad$ is at most $2^{-\Omega(n\sqrt{\ell})}.$

We now union bound over all size-$s$ circuits, of which there are $2^{O(s\log s)}$. Plugging in $s=\sqrt{\ell}$, we conclude 
\begin{align*}
    \Prx_{\bmcO}[\exists \text{ a size-$s$ $G$ s.t. }\bad]&\leq 2^{-\Omega(n\sqrt{\ell})}\cdot 2^{O(\sqrt \ell \log \ell)}.
\end{align*}
By \Cref{fact:code-exists}, $\ell=2^{\Theta(n^c)}$ where $c<1$ and so the above probability is at most $2^{-\Omega(n\sqrt{\ell})}$. The failure probability of~\Cref{thm:classical-lb-hitting} is therefore at most $2^{-\Omega(n)}$, dominated by that of~\Cref{claim:set-is-big}.

\section*{Acknowledgments}

We thank Scott Aaronson, Adam Bouland, Jordan Docter, and John Wright for helpful discussions.

Guy, Caleb, Carmen, and Li-Yang are supported by NSF awards 1942123, 2211237, 2224246, a Sloan Research Fellowship, and a Google Research Scholar Award. Guy is also supported by a Jane Street Graduate Research Fellowship and Carmen by an NSF GRFP. Jane is supported by NSF awards 2006664 and 310818 and an NSF GRFP.

\bibliographystyle{alpha}
\bibliography{ref}

\end{document}